\documentclass[12pt,letterpaper]{article}
\usepackage{amsmath,amsfonts,amsthm,amssymb,stmaryrd,relsize,bm}

\theoremstyle{plain}

\numberwithin{equation}{section}
\newtheorem{thm}{Theorem}[section]
\newtheorem{lem}[thm]{Lemma}

\newcommand{\complex}{{\mathbb C}}
\newcommand{\real}{{\mathbb R}}
\newcommand{\ascript}{{\mathcal A}}
\newcommand{\dscript}{{\mathcal D}}
\newcommand{\pscript}{{\mathcal P}}
\newcommand{\rscript}{{\mathcal R}}
\newcommand{\tscript}{{\mathcal T}}
\newcommand{\atilde}{\widetilde{a}}
\newcommand{\cbar}{\overline{c}}
\newcommand{\abar}{\overline{a}}
\newcommand{\pmathbf}{\mathbf{p}}
\newcommand{\qmathbf}{\mathbf{q}}
\newcommand{\sigmabm}{\bm{\sigma}}
\newcommand{\nablabm}{\bm{\nabla}}
\newcommand{\ctimes}{\mathrel{\mathlarger\cdot}}

\newcommand{\ab}[1]{\left|#1\right|}
\newcommand{\brac}[1]{\left\{#1\right\}}
\newcommand{\paren}[1]{\left(#1\right)}
\newcommand{\sqbrac}[1]{\left[#1\right]}
\newcommand{\elbows}[1]{{\left\langle#1\right\rangle}}
\newcommand{\floors}[1]{{\left\lfloor#1\right\rfloor}}
\newcommand{\ket}[1]{{\left|#1\right>}}

\begin{document}

\title{A UNIFIED APPROACH TO\\DISCRETE QUANTUM GRAVITY
}
\author{S. Gudder\\ Department of Mathematics\\
University of Denver\\ Denver, Colorado 80208, U.S.A.\\
sgudder@du.edu
}
\date{}
\maketitle

\begin{abstract} 
This paper is based on a covariant causal set (c-causet) approach to discrete quantum gravity. A c-causet is a partially ordered set $(x,<)$
that is invariant under labeling. We first consider the microscopic picture which describes the detailed structure of c-causets. The unique
labeling of a c-causet $x$ enables us to define a natural metric $d(a,b)$ between comparable vertices $a,b$ of $x$. The metric is then
employed to define geodesics and curvatures on $x$. We next consider the macroscopic picture which describes the growth process
$x \to y$ of c-causets. We propose that this process is governed by a quantum dynamics given by complex amplitudes. Denoting the
set of c-causets by $\pscript$ we show that the growth process $(\pscript ,\to )$ can be structured into a discrete 4-manifold.
This 4-manifold presents a unified approach to a discrete quantum gravity for which we define discrete analogues of Einstein's
field equations and Dirac's equation.
\end{abstract}

\section{Microscopic Picture}  
In this article we continue our work on a covariant causal set approach to discrete quantum gravity \cite{hen09,sor03,sur11}. For background and more details, we refer the reader to \cite{gud13,gud14}. We call a finite partially ordered set a \textit{causet} and interpret the order $a<b$ in a causet $x$ to mean that $b$ is in the causal future of $a$. We denote the cardinality of a causet $x$ by $\ab{x}$. If
$x$ and $y$ are causets with $\ab{y}=\ab{x}+1$ then $x$ \textit{produces} $y$ (written $x\to y$) if $y$ is obtained from $x$ by adjoining a single maximal element $a$ to $x$. If $x\to y$ we call $y$ an \textit{offspring} of $x$.

A \textit{labeling} for a causet $x$ is a bijection $\ell\colon x\to\brac{1,2,\ldots ,\ab{x}}$ such that $a,b\in x$ with $a<b$ implies that $\ell (a)<\ell (b)$. If $\ell$ is labeling for $x$, we call $x=(x,\ell )$ an $\ell$-\textit{causet}. Two $\ell$-causets $x$ and $y$ are \textit{isomorphic} if there exists a bijection
$\phi\colon x\to y$ such that $a<b$ in $x$ if and only if $\phi (a)<\phi (b)$ in $y$ and $\ell\sqbrac{\phi (a)}=\ell (a)$ for every $a\in x$. Isomorphic $\ell$-causets are considered identical as $\ell$-causets. We say that a causet is \textit{covariant} if it has a unique labeling (up to $\ell$-causet isomorphism) and call a covariant causet a $c$-\textit{causet}. We denote the set of a $c$-causets with cardinality $n$ by $\pscript _n$ and the set of all $c$-causets by $\pscript$. It is easy to show that any $x\in\pscript$ with $\ab{x}>1$ has a unique producer and that any $x\in\pscript$ has precisely two offspring \cite{gud14}. It follows that $\ab{\pscript _n}=2^{n-1}$, $n=1,2,\ldots\,$.

Two elements $a,b\in x$ are \textit{comparable} if $a<b$ or $b<a$. We say that $a$ is a \textit{parent} of $b$ and $b$ is a \textit{child} of $a$ if $a<b$ and there is no $c\in x$ with $a<c<b$. A \textit{path from} $a$ \textit{to} $b$ in $x$ is a sequence $a_1=a$, $a_2,\ldots a_{n-1}$, $a_n=b$ where $a_i$ is a parent of $a_{i+1}$, $i=1,\ldots ,n-1$. The \textit{height} $h(a)$ of $a\in x$ is the cardinality minus one of the longest path in $x$ that ends with $a$. If there are no such paths, then $h(a)=0$ by convention. It is shown in \cite{gud14} that a causet $x$ is covariant if and only if $a,b\in x$ are comparable whenever $a$ and $b$ have different heights. Notice that in any
$c$-causet, two elements with the same height cannot be comparable. If $x\in\pscript$ we call the sets
\begin{equation*}
S_j(x)=\brac{a\in x\colon h(a)=j}, j=0,1,2,\ldots
\end{equation*}
\textit{shells} and the sequence of integers $s_j(x)=\ab{S_j(x)}$, $j=0,1,2,\ldots$, is the \textit{shell sequence}. A $c$-causet is uniquely determined by its shell sequence and we think of $\brac{s_j(x)}$ as describing the ``shape'' or geometry of $x$ \cite{gud13,gud14}.

The tree $(\pscript ,\to )$ can be thought of as a growth model and an $x\in\pscript _n$ is a possible universe at step (time) $n$. An instantaneous universe $x\in\pscript _n$ grows one element at a time in one of two ways. To be specific, if $x\in\pscript _n$ has shell sequence $\paren{s_0(x),s_1(x),\ldots ,s_m(x)}$, then $x$ will grow to one of its two offspring $x\to x_0$, $x\to x_1$ where $x_0$ and $x_1$ have shell sequences
\begin{align*}
&\paren{s_0(x),s_1(x),\ldots ,s_m(x)+1}\\
&\paren{s_0(x),s_1(x),\ldots ,s_m(x),1}\\
\end{align*}
respectively. In this way, we can recursively order the $c$-causets in $\pscript$ by using the notation $x_{n,j}$, $n=1,2,\ldots$,
$j=0,1,2,\ldots ,2^{n-1}-1$ where $n=\ab{x_{n,j}}$. For example, in terms of their shell sequences we have:
\begin{align*}
x_{1,0}&=(1),x_{2,0}=(2),x_{2,1}=(1,1)\\
x_{3,0}&=(3),x_{3,1}=(2,1),x_{3,2}=(1,2),x_{3,3}=(1,1,1)\\
x_{4,0}&=(4),x_{4,1}=(3,1),x_{4,2}=(2,2),x_{4,3}=(2,1,1),x_{4,4}=(1,3)\\
x_{4,5}&=(1,2,1),x_{4,6}=(1,1,2),x_{4,7}=(1,1,1,1)
\end{align*}
 
In general, the $c$-causet $x_{n,j}$ has the two offspring $x_{n,j}\to x_{n+1,2j}$ and $x_{n,j}\to x_{n+1,2j+1}$, $n=1,2,\ldots$,
$j=0,1,2,\ldots ,2^{n-1}-1$. For example, $x_{3,2}\to x_{4,4}$ and $x_{3,2}\to x_{4,5}$ while $x_{3,3}\to x_{4,6}$ and $x_{3,3}\to x_{4,7}$. Conversely, for $n=2,3,\ldots$, $x_{n,j}$ has the unique producer $x_{n-1,\floors{j/2}}$ where $\floors{j/2}$ is the integer part of $j/2$. For example, $x_{5,14}$ had the producer $x_{4,7}$ and $x_{5,13}$ has the producer $x_{4,6}$. With the previously notation $\pscript =\brac{x_{n,j}}$ in place, we call $\brac{\pscript ,\to}$ a \textit{sequential growth process} (SGP).
 
 In the microscopic picture, we view a $c$-causet $x$ as a framework or scaffolding for a possible universe. The vertices of $x$ represent small cells that can be empty or occupied by a particle. The shell sequence that determines $x$ gives the geometry of the framework. In order to describe the universe, we would like to find out how particles move and which vertices they are likely to occupy. We accomplish this by introducing a distance or metric on $x$.
 
Let $x=\brac{a_1,a_2,\ldots ,a_n}\in\pscript _n$, where the subscript $i$ of $a_i$ is the label of the vertex. We can think of a path
\begin{equation}         
\label{eq11}
\gamma =a_{i_1}a_{i_2}\cdots a_{i_m}
\end{equation}
as a sequence in $x$ starting with $a_{i_1}$ and moving along successive shells until $a_{i_m}$ is reached. We define the \textit{length} of $\gamma$ by
\begin{equation}         
\label{eq12}
\ell (\gamma )=\sqbrac{\sum _{j=2}^m(i_j-i_{j-1})^2}^{1/2}
\end{equation}
Of course, there are a variety of definitions that one can give for the length of a path, but this is one of the simplest nontrivial choices. We shall compare \eqref{eq12} with another possible choice later in order to illustrate it's advantages. For $a,b\in x$ with $a<b$, a \textit{geodesic} from $a$ to $b$ is a path from $a$ to $b$ that has the shortest length. Clearly, if $a<b$, then there is at least one geodesic from $a$ to $b$. If $a,b\in x$ are comparable and $a<b$ say, then the \textit{distance} $d(a,b)$ is the length of a geodesic from $a$ to $b$. The next result shows that the triangle inequality holds when applicable so $d(a,b)$ has the most important property of a metric. A \textit{subpath} of a path $a_{i_1}a_{i_2}\cdots a_{i_m}$ is a subset of $\brac{a_{i_1},a_{i_2},\ldots ,a_{i_m}}$ that is again a path. The next result also shows that once we have a geodesic we can take subpaths to form other geodesics.

\begin{thm}       
\label{thm11}
{\rm{\cite{gud13}}}
{\rm{(i)}}\enspace If $a<c<b$, then $d(a,b)\le d(a,c)+d(c,b)$.
{\rm{(ii)}}\enspace A subpath of a geodesic is a geodesic.
\end{thm}

As we shall see, there may be more than one geodesic from $a$ to $b$ when $a<b$. In the remainder of this section, we shall refer to a vertex by its label. If there are $j$ geodesics from vertex 1 to vertex $n$, we define the \textit{curvature} $K(n)$ at $n$ to be $K(n)=j-1$. One might argue that the curvature should be a local property and should not depend so heavily on vertex 1 which could be a considerable distance away. However, if there are a lot of geodesics from 1 to $n$, then by Theorem~\ref{thm11}(ii), there are also a lot of geodesics from other vertices to $n$. Thus, the definition of curvature is not so dependent on the initial vertex 1 as it first appears. Assuming that particles tend to move along geodesics, we see that $K(n)$ gives a measure of the tendency for vertex $n$ to be occupied.

We now give an example that appeared in \cite{gud13} and we refer the reader to that reference for details and other examples. Let $x$ be the $c$-causet with shell sequence (1,2,3,4,5,4,3,2,1). Separating shells by semicolons, the labels of the vertices become
\begin{equation*}
(1;2,3;4,5,6;7,8,9,10;11,12,13,14,15;16,17,18,19;20,21,22;23,24;25)
\end{equation*}
The $c$-causet $x$ represents a toy universe that expands uniformly and then contracts uniformly. The following table summarizes the curvatures for vertices of $x$.
\bigskip

\begin{tabular}{|c|c|c|c|c|c|c|c|c|c|c|c|c|c|c|}
\hline
$i$&1&2&3&4&5&6&7&8&9&10&11&12&13&14\\
\hline
$K(i)$&$-1$&0&0&1&0&0&0&1&0&0&2&0&1&0\\
\hline\noalign{\bigskip}
\end{tabular}

\begin{tabular}{|c|c|c|c|c|c|c|c|c|c|c|c|c|}
\hline
$i$&15&16&17&18&19&20&21&22&23&24&25\\
\hline
$K(i)$&$0$&2&2&0&1&5&3&1&5&9&5\\
\hline\noalign{\medskip}
\multicolumn{12}{c}%
{\textbf{Table 1}}\\
\end{tabular}

For illustrative purposes, we now give an example of a different distance function. For a path $\gamma$ given by \eqref{eq11} define the \textit{length} $\ell _1$ by
\begin{equation}         
\label{eq13}
\ell _1(\gamma )=\max\brac{\ab{i_j-i_{j-1}}\colon j=2,\ldots ,m}
\end{equation}
We now define 1-\textit{geodesics}, \textit{distance} $d_1(a,b)$ and \textit{curvature} $K_1$ as before with $\ell (\gamma )$ replaced by $\ell _1(\gamma )$. Although $d_1$ satisfies the triangle inequality of Theorem~\ref{thm11}(i), it gives a rather course distance measure compared to $d$. One way of seeing this is that Theorem~\ref{thm11}(ii) does not hold for $d_1$ as the next example shows.

Let $x$ be the $c$-causet with shell sequence (1,2,3,4) and vertex labels (1;2,3;4,5,6;7,8,9,10). The following two tables summarize the distances and curvatures for the vertices of $x$.
\bigskip

{\hskip 6pc
\begin{tabular}{|c|c|c|c|c|c|c|c|c|c|}
\hline
$i$&2&3&4&5&6&7&8&9&10\\
\hline
$d_1(1,i)$&1&2&2&2&3&2&3&3&4\\
\hline\noalign{\medskip}
\multicolumn{10}{c}%
{\textbf{Table 2}}\\
\end{tabular}}
\bigskip

{\hskip 6pc
\begin{tabular}{|c|c|c|c|c|c|c|c|c|c|c|}
\hline
$i$&1&2&3&4&5&6&7&8&9&10\\
\hline
$K_1(i)$&$-1$&0&0&1&0&0&0&2&0&1\\
\hline\noalign{\medskip}
\multicolumn{11}{c}%
{\textbf{Table 3}}\\
\end{tabular}}
\bigskip

\noindent Notice that the path 1-2-5-8 is a 1-geodesic, but the path 1-2-5 is not which shows that Theorem~\ref{thm11}(ii) does not hold for $d_1$. Another example is that 1-2-6-10 is a 1-geodesic, but 1-2-6 is not.

\section{Macroscopic Picture} 
In general relativity theory it is postulated that the mass-energy distribution determines the curvature, while in our microscopic picture we assume that it is the other way around. That is, the curvature determines the mass distribution and the curvature is given by the geometry (shell sequence). We are now confronted with the question: What determines the shell sequence of our particular universe? To study this question, the present section studies the macroscopic picture. This picture describes the evolution of a universe as a quantum sequential growth process. In such a process, the probabilities of competing evolutions are determined by quantum amplitudes. Moreover, we shall see the emergence of a discrete 4-manifold.

In \cite{gud13} we gave a method for constructing a discrete 4-manifold from the SGP $(\pscript ,\to)$. This method was based on forming
``twin $c$-causets.'' In this article we find it convenient to employ a related but different method based on incident edge pairs. If $x\to y$ we call $xy$ an \textit{edge} in $\pscript$. Two edges of the form $e=xy$ and $f=yz$ are said to be \textit{incident} and we call the pair $(e,f)$ of incident edges a \textit{direction}. The direction $(e,f)$ can also be described by the vertices $(x,y,z)$ where $x\to y\to z$. Starting at any $x_{n,j}\in\pscript$ we have the four directions given by
\begin{align*}
d_{n,j}^1&=(x_{n,j},x_{n+1,2j},x_{n+2,4j})\\
d_{n,j}^2&=(x_{n,j},x_{n+1,2j},x_{n+2,4j+1})\\
d_{n,j}^3&=(x_{n,j},x_{n+1,2j+1},x_{n+2,4j+2})\\
d_{n,j}^4&=(x_{n,j},x_{n+1,2j+1},x_{n+2,4j+3})\\
\end{align*}

We denote the set of directions by
\begin{equation*}
D=\brac{d_{n,j}^k\colon n=1,2,\ldots ,j=0,1,\ldots ,2^{n-1}-1,k=1,2,3,4}
\end{equation*}
Two directions $(e,f)$, $(e_1,f_1)$ are \textit{incident} if $f$ and $e_1$ are incident. A \textit{direction path} $\omega$ in $\pscript$ is a
sequence of successively incident directions and a \textit{direction} $n$-\textit{path} in $\pscript$ is a finite sequence of $n$ successively incident directions beginning at one of the \textit{initial vertices} $x_{1,0}$, $x_{2,0}$ or $x_{2,1}$. Specifically, a direction $n$-path has one of the forms
\begin{align}             
\label{eq21}      
d_{1,0}^{k_1}d_{3,j_3}^{k_3}&\cdots d_{2n-1,j_{2n-1}}^{k_{2n-1}}\\
\label{eq22}      
d_{2,0}^{k_2}d_{4,j_4}^{k_3}&\cdots d_{2n,j_{2n}}^{k_{2n}}\\
\label{eq23}      
d_{2,1}^{k_2}d_{4,j_4}^{k_4}&\cdots d_{2n,j_{2n}}^{k_{2n}}
\end{align}
where $k_i\in\brac{1,2,3,4}$ and each direction is incident to the direction that follows it. We say that \eqref{eq21}, \eqref{eq22}, \eqref{eq23} have \textit{initial vertices} $x_{1,0},x_{2,0},x_{2,1}$, respectively and have \textit{final vertices} given by the last vertices of their directions, respectively. For example, the unique direction 2-path with initial vertex $x_{1,0}$ and final vertex $x_{5,7}$ is
$d_{1,0}^2d_{3,1}^4$. Notice that every \textit{odd vertex} $x_{2n-1,j_{2n-1}}$ is the final vertex of a unique direction $n$-path with initial vertex $x_{1,0}$ and every \textit{even vertex} $x_{2n,j_{2n}}$ is the final vertex of a unique direction $n$-path with initial vertex $x_{2,0}$ or $x_{2,1}$. A direction path is like one of the direction $n$-paths \eqref{eq21}, \eqref{eq22}, \eqref{eq23} except it does not terminate.

We denote the set of direction paths in $\pscript$ by $\Omega$ and the set of direction $n$-paths by $\Omega _n$. We interpret a direction path as a completed universe or history of an evolved universe. A direction path or $n$-path $\omega$ \textit{contains} $x\in\pscript$ if there is a direction of
$\omega$ that has the form $(x,y,z)$. We then write $x\in\omega$.

A \textit{transition amplitude} is a map $\atilde\colon D\to\complex$ satisfying $\sum _{k=1}^4\atilde (d_{n,j}^k)=1$ for all $n,j$. Corresponding to
$\atilde$ we define the \textit{amplitude} of a direction $n$-path $\omega =\omega _1\omega _2\cdots\omega _n\in\Omega _n$ to be
$a(\omega )=\atilde (\omega _1)\atilde (\omega _2)\cdots\atilde (\omega _n)$ if the initial vertex of $\omega$ is $x_{1,0}$ and
$a(\omega )=(1/2)\atilde (\omega _1)\atilde (\omega _2)\cdots\atilde (\omega _n)$ if the initial vertex of $\omega$ is $x_{2,0}$ or $x_{2,1}$.
The \textit{amplitude} of a set $A\subseteq\Omega _n$ is $a(Q)=\sum\brac{a(\omega )\colon\omega\in A}$. Notice that $a(\Omega _n)=1$. The \textit{amplitude} $a(x_{n,j})$ for $x_{n,j}\in\pscript _n$, $n\ge 3$ is $a(\omega )$ where $\omega$ is the unique finite direction path that has
$x_{n,j}$ as its final vertex. It follows that $\sum\brac{a(x)\colon x\in\pscript _n}=1$. We call $c_{n,j}^k=\atilde (d_{n,j}^k)$
\textit{coupling constants} and note that $\sum _{k=1}^4c_{n,j}^k=1$ for all $n,j$. We define $a(x_{n,0})=1$, $a(x_{2,0})=a(x_{2,1})=1/2$.

The $q$-\textit{measure} of a set $A\subseteq\Omega _n$ corresponding to $\atilde$ is defined by $\mu _n(A)=\ab{a(A)}^2$ \cite{sor94}. In particular, $\mu _n(\omega )=\ab{a(\omega )}^2$, $\omega\in\Omega _n$ and $\mu _n(x_{n,j})=\ab{a(x_{n,j})}^2$. Letting
$\ascript _n=2^{\Omega _n}$ be the power set on $\Omega _n$ we have that $(\Omega _n,\ascript _n)$ is a measurable space and we call $(\Omega _n,\ascript _n,\mu _n)$ a $q$-\textit{measure space}. We interpret $\mu _n(A)$ as the quantum propensity of the event
$A\in\ascript _n$. It is believed that once the coupling constants and hence the $q$-measures $\mu _n$ are known, then certain direction $n$-paths and $c$-causets $x_{n,j}$ will have dominate propensities. In this way we will determine dominate geometries for the microscopic picture of our particular universe. Because of quantum interference the $q$-measure $\mu _n$ is not a measure on the $\sigma$-algebra $\ascript _n$, in general. This is because the additivity condition $\mu _n(A\cup B)=\mu _n(A)+\mu _n(B)$ whenever $A\cap B=\emptyset$ need not hold. Of course, we do have that $\mu _n(\Omega _n)=1$ and $\mu _n(A)\ge 0$ for all $A\in\ascript$.

We refer the reader to \cite{gud13} for an example of a transition amplitude $\atilde\colon D\to\complex$ that may have physical significance. In this situation, the $c$-causets and direction paths with highest propensity lie toward the ``middle'' of the process $(\pscript,\to )$. For example, the
$c$-causets of highest propensity are those $x_{n,j}$ with $j=2^{n-2}$ and the propensities decrease to zero for large $n$ as $j$ gets smaller and larger than $2^{n-2}$

\section{Covariant Difference Operators} 
Viewing the directions $d_{n,j}^k$, $k=1,2,3,4$ as ``tangent vectors'' at the vertex $x_{n,j}$ we see the emergence of a discrete 4-manifold. We now carry this analogy further by defining covariant difference operators.

Let $H=L_2(\pscript )$ be the Hilbert space of square summable complex-valued functions on $\pscript$ with the usual inner product
\begin{equation*}
\elbows{f,g}=\sum _{x\in\pscript}\overline{f(x)}g(x)
\end{equation*}
We define the \textit{covariant difference operators} $\nabla _k$, $k=1,2,3,4$ by
\begin{equation}         
\label{eq31}
\nabla _kf(x_{n,j})=f(x_{n+2,4j+k-1})-c_{n+2,4j+k-1}^kf(x_{n,j})
\end{equation}
The operator $\nabla _k$ can be considered to be the difference operator in the direction $k$, $k=1,2,3,4$. This is a slight variation of the difference operators considered in \cite{gud12,gud13}.

For $p=(p_1,p_2,p_3,p_4)\in\real ^4$, define the function $w\colon\real ^4\times\pscript\to\complex$ recursively by
\begin{equation}         
\label{eq32}
w(p,x_{n+2,4j+k-1})=
\begin{cases}(c_{n+2,4j+k-1}^k+ip_k)w(p,x_{n,j})&\text{if\ $k=1,2,3$}\\
  (c_{n+2,4j+k-1}^k-ip_k)w(p,x_{n,j})&\text{if\ $k=4$}
\end{cases}
\end{equation}
where $i=\sqrt{-1}$. The values $w(p,x_{1,0})$, $w(p,x_{2,0})$, $w(p,x_{2,1})$ are arbitrary and are the \textit{initial conditions}. For fixed $p$, $w(p,x_{n,j})$ corresponds to a \textit{discrete plane wave} in the ``direction'' $p$. In general, $w(p,\ctimes )\notin H$ except for certain values of $p$ depending on the coupling constants.
\medskip

\noindent\textbf{Example.}\enspace If we define the initial conditions
\begin{equation*}
w(p,x_{1,0})=w(p,x_{2,0})=w(p,x_{2,1})=1
\end{equation*}
we have that
\begin{align*}
w(p,x_{3,j})&=c_{3,j}^{j+1}+ip_{j+1},\ j=0,1,2\\
w(p,x_{4,j})&=c_{4,j}^{j\!\!\!\!\pmod{4}+1}+ip_{j\!\!\!\!\pmod{4}+1},\ j=0,1,2,4,5,6\\  
w(p,x_{5,j})&=(c_{5,j}^{j+1}+ip_1)(c_{3,0}^1+ip_1),\ j=0,1,2\\
w(p,x_{5,j})&=c_{5,j}^{j\!\!\!\!\pmod{4}+1}+ip_{j\!\!\!\!\pmod{4}+1}(c_{3,1}^2+ip_2),\ j=4,5,6\\
w(p,x_{5,j})&=c_{4,j}^{j\!\!\!\!\pmod{4}+1}+ip_{j\!\!\!\!\pmod{4}+1}(c_{3,2}^3+ip_3),\ j=8,9,10\\
w(p,x_{5,j})&=c_{4,j}^{j\!\!\!\!\pmod{4}+1}+ip_{j\!\!\!\!\pmod{4}+1}(c_{3,3}^4+ip_4),\ j=12,13,14\\
\end{align*}
We have omitted the values of $w(p,x_{n,j})$ for $j=3\pmod{4}$. These are given by
\begin{align*}
w(p,x_{3,3})&=c_{3,3}^4-ip_4\\
w(p,x_{4,3})&=c_{34,3}^4-ip_4\\
w(p,x_{4,7})&=c_{4,7}^4-ip_4\\
w(p,x_{5,3})&=(c_{5,3}^4-ip_4)(c_{3,0}^1+ip_2)\\
w(p,x_{5,7})&=(c_{5,7}^4-ip_4)(c_{3,1}^2+ip_2)\\
w(p,x_{5,11})&=(c_{5,11}^4-ip_4)(c_{3,2}^3+ip_3)\\
w(p,x_{5,15})&=(c_{5,15}^4-ip_4)(c_{3,3}^4+ip_4)\\
\end{align*}
\medskip

In general, $\nabla _k$ is an unbounded operator and we denote its domain by $\dscript (\nabla _k)$ $k=1,2,3,4$. The next result shows that if $\nabla _kw(p,\ctimes )$ is defined, then $w(p,\ctimes )$ is a simultaneous eigenvector of $\nabla _k$, $k=1,2,3,4$.

\begin{lem}       
\label{thm31}
If $w(p,\ctimes )\in\bigcap\limits _{k=1}^4\dscript (\nabla _k)$, then
\begin{equation*}
\nabla _kw(p,x_{n,j})=
\begin{cases}ip_kw(p,x_{n,j})&\text{if\ $k=1,2,3$}\\
  -ip_kw(p,x_{n,j})&\text{if\ $k=4$}
\end{cases}
\end{equation*}
\end{lem}
\begin{proof}
By \eqref{eq31} and \eqref{eq32} we have
\begin{align*}
\nabla _kw(p,x_{n,j})&=w(p,x_{n+2,4j+k-1})-c_{n+2,4j+k-1}^kw(p,x_{n,j})\\
  &=(c_{n+2,4j+k-1}^k\pm ip_k)-c_{n+2,4j+k-1}w(p,x_{n,j})\\
&=\begin{cases}ip_kw(p,x_{n,j})&\text{if\ $k=1,2,3$}\\
  -ip_kw(p,x_{n,j})&\text{if\ $k=4$}\hskip 12pc\qedhere
\end{cases}
\end{align*}
\end{proof}

If $\omega =\omega _1\omega _2\cdots\in\Omega$ and $x_{n,j}\in\omega$, then $\omega _m=(x_{n,j},y,z)$ for some integer $m$. Now for some $m',j',k'$ we have $\omega _m=d_{m',j'}^{k'}$, and we write $k(\omega ,x_{n,j})=k'$. We think of $k(\omega ,x_{n,j})$ as specifying the direction of $\omega$ at the vertex $x_{n,j}$. Corresponding to the amplitude $a$ and an arbitrary $x_{n,j}\in\pscript$ we define
\begin{equation*}
a_\omega (x_{n,j})=
\begin{cases}a (x_{n,j})&\text{if\ $x_{n,j}\in\omega$}\\
  0&\text{if\ $x_{n,j}\notin\omega$}
\end{cases}
\end{equation*}
The $\omega$-\textit{covariant difference operator} is the operator $\nabla _\omega$ on $H$ given by
\begin{equation}         
\label{eq33}
\nabla _\omega f(x_{n,j})=a_\omega (x_{n,j})\nabla _{k(\omega ,x_{n,j})}f(x_{n,j})
\end{equation}
We then have by \eqref{eq31} that
\begin{align}         
\label{eq34}
\nabla _\omega f(x_{n,j})&=a_\omega (x_{n,j})f(x_{n+2,4j+k(\omega ,x_{n,j})-1})\notag\\
&\quad -a_\omega (x_{n+2,4j+k(\omega ,x_{n,j})-1})f(x_{n,j})
\end{align}
If $a_\omega\in H$, then it follows from \eqref{eq34} that $\nabla _\omega a_\omega =0$ which is why $\nabla _k$ and $\nabla _\omega$ are called covariant.

We now extend this formalism to functions of two variables. Let $K=H\otimes H$ which we can identify with $L_2(\pscript\times\pscript )$. For $k,k'=1,2,3,4$, we define the \textit{covariant bidifference operators} $\nabla _{k,k'}$ with domains in $K$ by
\begin{align*}
\nabla _{k,k'} f(x_{n,j},x_{n'j'})&=f(x_{n+2,4j+k-1},x_{n'+2,4j'+k'-1})\\
&\quad -\cbar _{n+2,4j+k-1}^kc_{n'+2,4j'+k'-1}^{k'}f(x_{n,j},x_{n',j'})
\end{align*}
For $\omega ,\omega '\in\Omega$, the $\omega\omega '$-\textit{covariant bidifference operator} is given by
\begin{equation*}
\nabla _{\omega ,\omega '}f(x_{n,j},x_{n'j'})=\abar _\omega (x_{n,j})a_{\omega '}(x_{n',j'})
  \nabla _{k(\omega ,x_{n,j}),k(\omega ',x_{n',j'})}f(x_{n,j},x_{n'j'})
\end{equation*}
Again, $\nabla _{\omega,\omega '}$ is called covariant because we have $\nabla _{\omega ,\omega '}\abar _\omega a_{\omega '}=0$

\section{Discrete Dirac and Einstein Equations} 
This section shows that we can employ the covariant difference operators presented in Section~3 to construct discrete analogues of Dirac's and Einstein's equations. Letting $c$ be the speed of light in a vacuum, we use units in which $c=\hbar=1$. In $\real ^4$ we employ the indefinite inner product
\begin{equation*}
p\cdot q=-p_1q_1-p_2q_2-p_3q_3+p_4q_4
\end{equation*}
and we let $\sigma _k$, $k=1,2,3,4$, be the Pauli matrices
\begin{equation*}
\sigma _1=\begin{bmatrix}0&1\\1&0\end{bmatrix},\quad
\sigma _2=\begin{bmatrix}0&-i\\i&0\end{bmatrix},\quad
\sigma _3=\begin{bmatrix}1&0\\0&-1\end{bmatrix},\quad
\sigma _4=\begin{bmatrix}1&0\\0&1\end{bmatrix},\quad
\end{equation*}
Defining $\sigma =(\sigma _1,\sigma _2,\sigma _3,\sigma _4)$ and $\nabla =(\nabla _1,\nabla _2,\nabla _3,\nabla _4)$ we have that
\begin{equation*}
\sigma\cdot\nabla =-\sigma _1\nabla _1-\sigma _2\nabla _2-\sigma _3\nabla _3+\sigma _4\nabla _4
\end{equation*}
The \textit{discrete Weyl equation} is defined by
\begin{equation}         
\label{eq41}
\sigma\cdot\nabla \phi (x_{n,j})=0
\end{equation}
where $\phi\colon\pscript\to\complex ^2$ is a two-component function $\phi =(\phi _1,\phi _2)$. Writing \eqref{eq41} in full gives
\begin{equation}         
\label{eq42}
\begin{bmatrix}\noalign{\smallskip}
-\nabla _3+\nabla _4&-\nabla _1+i\nabla _2\\-\nabla _1-i\nabla _2&\nabla _3+\nabla _4\\\noalign{\smallskip}\end{bmatrix}\quad
\begin{bmatrix}\noalign{\smallskip}\phi _1(x_{n,j})\\\phi _2(x_{n,j})\\\noalign{\smallskip}\end{bmatrix}=0
\end{equation}

For $p\in\real ^4$, let $u(p)$ be a two-component function $u\colon\real ^4\to\complex ^2$ satisfying
\begin{equation}         
\label{eq43}
\sigma\cdot p\ u(p)=0
\end{equation}
Writing \eqref{eq43} in full gives
\begin{equation}         
\label{eq44}
\begin{bmatrix}\noalign{\smallskip}
-p_3+p_4&-p _1+ip _2\\-p _1-ip _2&p _3+p _4\\\noalign{\smallskip}\end{bmatrix}\quad
\begin{bmatrix}\noalign{\smallskip}u_1(p)\\ u_2(p)\\\noalign{\smallskip}\end{bmatrix}=0
\end{equation}
If there is a nonzero solution of \eqref{eq43}, the determinate of the matrix in \eqref{eq44} must be zero. Hence,
\begin{equation}         
\label{eq45}
p\cdot p=-p_1^2-p_2^2-p_3^2+p_4^2=0
\end{equation}
Moreover, solutions of \eqref{eq43} have the form $u_1(p)=1$
\begin{equation*}
u_2(p)=(p_3-p_4)/(-p_1+ip_2)
\end{equation*}

\begin{lem}       
\label{thm41}
Discrete plane wave solutions of \eqref{eq41} have the form
\begin{equation}         
\label{eq46}
\phi (x_{n,j})=u(p)w(p,x_{n,j})
\end{equation}
when $p$ satisfies \eqref{eq45} and $w(p,\ctimes)\in\bigcap\limits _{k=1}^4\dscript (\nabla _k)$.
\end{lem}
\begin{proof}
Letting $\phi (x_{n,j})$ be defined by \eqref{eq46} we have by \eqref{eq43} that
\begin{align*}
\sigma\cdot\nabla\phi (x_{n,j})&=\sigma\cdot\nabla u(p)w(p_,x_{n,j})\\
&=\begin{bmatrix}\noalign{\smallskip}
-\nabla _3+\nabla _4&-\nabla _1+i\nabla _2\\-\nabla _1-i\nabla _2&\nabla _3+\nabla _4\\\noalign{\smallskip}\end{bmatrix}\quad
\begin{bmatrix}\noalign{\smallskip}
u_1(p)w(p,x_{n,j})\\u_2(p)w(p,x_{n,j})\\\noalign{\smallskip}\end{bmatrix}\\
\noalign{\bigskip}
&=\begin{bmatrix}\noalign{\smallskip}
(-p_3+p_4)w(p,x_{n,j})u_1(p)+(-p_1+ip_2)w(p,x_{n,j})u_2(p)\\
(-p_1-ip_2)w(p,x_{n,j})u_1(p)+(p_3+p_4)w(p,x_{n,j})u_2(p)\\\noalign{\smallskip}\end{bmatrix}\\
\noalign{\smallskip}
&=w(p,x_{n,j})\sigma\cdot p u(p)=0\hskip 12pc\qedhere
\end{align*}
\end{proof}

Define the $4\times 4$ \textit{gamma matrices} by
\begin{equation*}
\gamma _4=\begin{bmatrix}\sigma _4&0\\0&-\sigma _4\end{bmatrix},\quad
\gamma _k=\begin{bmatrix}0&-\sigma _k\\\sigma _k&0\end{bmatrix},\quad k=1,2,3
\end{equation*}
We define the \textit{discrete free Dirac equation} by
\begin{equation}         
\label{eq47}
(i\gamma\cdot\nabla -m)\phi (x_{n,j})=0
\end{equation}
where $m>0$ and $\phi\colon\pscript\to\complex ^4$ is a four-component function $\phi =(\phi _1,\phi _2,\phi _3,\phi _4)$ called a
\textit{spinor}. Notice that when $m=0$, \eqref{eq47} essentially reduces to \eqref{eq41}. Equation~\eqref{eq47} says that $\phi$ is an eigenvector of the \textit{discrete free Dirac operator} $i\gamma\cdot\nabla$ with eigenvalue $m$. We now proceed in the usual way to find solutions of \eqref{eq47}. Although the method is standard, for completeness we give some details.

For $\pmathbf ,\qmathbf\in\real ^3$ we use the usual inner product $\pmathbf\cdot\qmathbf =p_1q_1+p_2q_2+p_3q_3$ and norm
$\|\pmathbf\|=(\pmathbf\cdot\pmathbf )^{1/2}$. For $\pmathbf =(p_1,p_2,p_3)\in\real ^3$ define $p_4=\sqrt{\|\pmathbf\|^2+m^2\,}$ and let 
$\ket{0}$, $\ket{1}$ be the qubits
\begin{equation*}
\ket{0}=\begin{bmatrix}1\\0\end{bmatrix},\quad
\ket{1}=\begin{bmatrix}0\\1\end{bmatrix}
\end{equation*}
For $s=0,1$, $\pmathbf\in\real ^3$, $p_4=\sqrt{\|\pmathbf\|^2+m^2\,}$, $\bm{\sigma} =(\sigma _1,\sigma _2,\sigma _3)$ define the
four-vectors
\begin{equation*}
u(\pmathbf ,s)
=\begin{bmatrix}\ket{s}\\\noalign{\smallskip}\frac{\sigmabm\cdot\pmathbf}{p_4+m}\ket{s}\\\noalign{\smallskip}\end{bmatrix},\quad
v(\pmathbf ,s)
=\begin{bmatrix}\noalign{\smallskip}\frac{\sigmabm\cdot\pmathbf}{p_4+m}(-i\sigma _2\ket{s})\\\noalign{\smallskip}
-i\sigma _2\ket{s}\end{bmatrix}
\end{equation*}
and note that $-i\sigma _2\ket{0}=\ket{1}$ and $-i\sigma _2\ket{1}=-\ket{0}$

\begin{thm}       
\label{thm42}
Letting $p=(p_1,p_2,p_3,p_4)=(\pmathbf ,p_4)$ where $p_4=\sqrt{\|\pmathbf\|^2+m^2\,}$, discrete plane wave solutions of \eqref{eq47} are
\begin{align*}
\psi _{p,s}^{(+)}(x_{n,j})&=u(\pmathbf ,s)w(p, x_{n,j}),\quad s=0,1\\
\intertext{and}
\psi _{-p,-s}^{(-)}(x_{n,j})&=v(\pmathbf ,s)w(-p,x_{n,j}),\quad s=0,1
\end{align*}
\end{thm}
\begin{proof}
Defining $\nablabm =(\nabla _1,\nabla _2,\nabla _3)$ we can write the Dirac operator as
\begin{align}         
\label{eq48}
i\gamma\cdot\nabla&=-i\gamma _1\nabla _1-i\gamma _2\nabla _2-i\gamma _3\nabla _3+i\gamma _4\nabla _4\notag\\
&=i\begin{bmatrix}\nabla _4&\sigmabm\cdot\nablabm\\-\sigmabm\cdot\nablabm&-\nabla _4\end{bmatrix}
\end{align}
Consider the case $s=0$. We have that
\begin{align*}
\sigmabm\cdot\pmathbf\ket{0}&=(p_1\sigma _1+p_2\sigma _2+p_3\sigma _3)\ket{0}\\
   &=\begin{bmatrix}p_3&p_1-ip_2\\p_1+ip_2&-p_3\end{bmatrix}\ \begin{bmatrix}1\\0\end{bmatrix}
   =\begin{bmatrix}p_3\\p_1+ip_2\end{bmatrix}
\end{align*}
Hence,
\begin{align}         
\label{eq49}
\sigmabm\cdot\nablabm\frac{\sigmabm\cdot\pmathbf}{p_4+m}&\ket{0}w(p,x_{n,j})
=\frac{1}{p_4+m}\sigmabm\cdot\nablabm\begin{bmatrix}p_3\\p_1+ip_2\end{bmatrix}w(p,x_{n,j})\notag\\
&=\frac{1}{p_4+m}\begin{bmatrix}\nabla_3&\nabla _1-i\nabla _2\\\nabla_1+i\nabla_2&-\nabla _3\end{bmatrix}\ 
\begin{bmatrix}p_3\\p_1+ip_2\end{bmatrix}w(p,x_{n,j})\notag\\
&=\frac{1}{p_4+m}\begin{bmatrix}ip_3^2&(p_1+ip_2)(ip_1+p_2)\\p_3(ip_1-p_2)&-(p_1+ip_2)p_3\end{bmatrix}w(p,x_{n,j})\notag\\
&=\frac{1}{p_4+m}\begin{bmatrix}i\|\pmathbf\|^2\\0\end{bmatrix}w(p,x_{n,j})
\end{align}
By \eqref{eq48} and \eqref{eq49} we have that
\begin{align*}
i\gamma\cdot\nabla\psi _{p,0}^{(+)}(x_{n,j})
 &=i\begin{bmatrix}\nabla _4&\sigmabm\cdot\nablabm\\-\sigmabm\cdot\nablabm&-\nabla _4\end{bmatrix}u(\pmathbf ,0)w(p,x_{n,j})\\
 &=i\begin{bmatrix}\nabla _4&\sigmabm\cdot\nablabm\\-\sigmabm\cdot\nablabm&-\nabla _4\end{bmatrix}\ 
 \begin{bmatrix}\ket{0}\\\frac{\sigmabm\cdot\pmathbf}{p_4+m}\ket{0}\\\noalign{\smallskip}\end{bmatrix}w(p,x_{n,j})\\
 &=i\begin{bmatrix}\noalign{\smallskip}-ip_4\ket{0}+\frac{i}{p_4+m}\|\pmathbf\|^2\ket{0}\\\noalign{\smallskip}
{\begin{bmatrix}-ip_3\\-ip_1+p_2\end{bmatrix}}+\frac{ip_4}{p_4+m}{\begin{bmatrix}p_3\\p_1+ip_2\end{bmatrix}}
\\\noalign{\smallskip}\end{bmatrix}w(p,x_{n,j})\\
 &=i\begin{bmatrix}\noalign{\smallskip}\paren{p_4+\frac{\|\pmathbf\|^2}{p_4+m}}\ket{0}\\\noalign{\smallskip}
 \paren{1-\frac{p_4}{p_4+m}}\sigmabm\cdot\pmathbf\ket{0}\\\noalign{\smallskip}\end{bmatrix}w(p,x_{n,j})\\
 &=\begin{bmatrix}\noalign{\smallskip}
 \frac{p_4^2+mp_4-\|\pmathbf\|^2}{p_4+m}\ket{0}\\\noalign{\smallskip}\frac{m}{p_4+m}\sigmabm\cdot\pmathbf\ket{0}\\
 \noalign{\smallskip}\end{bmatrix}w(p,x_{n,j})\\
 &=m\begin{bmatrix}\ket{0}\\\frac{\sigmabm\cdot\pmathbf}{p_4+m}\ket{0}\\\noalign{\smallskip}\end{bmatrix}
 w(p,x_{n,j})=m\psi _{p,0}^{(+)}(x_{n,j})
\end{align*}
the other cases are similar.
\end{proof}

Finally, we consider a discrete analogue of Einstein's field equations. The \textit{global curvature operator} is defined by
$\rscript _{\omega ,\omega '}=\nabla _{\omega ,\omega '}-\nabla _{\omega ',\omega}$. The term global is used to distinguish this from the local curvature we defined in the microscopic picture. We then have that
\begin{align*}
\rscript _{\omega ,\omega '}f(x_{n,j},x_{n',j'})
&=\abar _\omega (x_{n,j})a _{\omega '}(x_{n',j'})\nabla _{k(\omega ,x_{n,j}),k(\omega ',x_{n,j})}f(x_{n,j},x_{n',j'})\\
&\quad -\abar _{\omega '}(x_{n,j})a_\omega (x_{n',j'})\nabla _{k(\omega ',x_{n',j'}),k(\omega ,x_{n,j})}f(x_{n,j},x_{n',j'})
\end{align*}
Define the operators $\dscript _{k,k'}$ with domains in $K$ by
\begin{align*}
\dscript _{k,k'}f(x_{n,j},x_{n',j'})
&=\abar _\omega (x_{n,j})a _{\omega '}(x_{n',j'})f(x_{n+2,4j+k-1},x_{n'+2,4j'+k'-1})\\
&\quad -\abar _{\omega '}(x_{n,j})a_\omega (x_{n',j'})f(x_{n'+2,4j'+k'-1},x_{n+2,4j+k-1})
\end{align*}
and the operators $\tscript _{k,k'}$ with domains in $K$ by
\begin{align*}
\tscript _{k,k'}f(x_{n,j},x_{n',j'})
&=\left[\abar _{\omega '}(x_{n+2,4j+k-1})a_\omega (x_{n'+2,4j'+k'-1})\right.\\
&\quad \left.-\abar _\omega (x_{n+2,4j+k-1})a_{\omega '}(x_{n'+2,4j'+k'-1})\right]f(x_{n,j},x_{n',j'})
\end{align*}
If we define $\dscript _{\omega ,\omega '}$ and $\tscript _{\omega ,\omega '}$ by
\begin{align*}
\dscript _{\omega ,\omega '}f(x_{n,j},x_{n',j'})&=\dscript _{k(\omega ,x_{n,j}),k(\omega ',x_{n',j'})}f(x_{n,j},x_{n',j'})\\
\intertext{and}
\tscript _{\omega ,\omega '}f(x_{n,j},x_{n',j'})&=\tscript _{k(\omega ,x_{n,j}),k(\omega ',x_{n',j'})}f(x_{n,j},x_{n',j'})
\end{align*}
then it is easy to check that
\begin{equation}         
\label{eq410}
\rscript _{\omega ,\omega '}=\dscript _{\omega ,\omega '}+\tscript _{\omega ,\omega '}
\end{equation}
We call \eqref{eq410} the \textit{discrete Einstein equations}. For a further discussion of these equations, we refer the reader to \cite{gud12}.

\end{document}